\documentclass[runningheads,a4paper,USenglish]{llncs}
\usepackage[colorlinks=true,pdfborder=0]{hyperref}
\usepackage{xspace}
\usepackage{amsmath,amsfonts,amssymb}
\usepackage{esvect}
\usepackage{cleveref}
\usepackage{thm-restate}
\usepackage[font=small, labelfont=bf]{caption}
\usepackage[font=small, labelfont=normal]{subcaption}

\title{Coloring Mixed and Directional Interval Graphs%
	\thanks{Work on this problem was initiated at the HOMONOLO Workshop 2021
	in Nov\'a Louka, Czech Republic. G.G.\ is partially
    supported by the National Science Center of Poland under grant no.\
    2019/35/B/ST6/02472.}}

\author{Grzegorz~Gutowski\inst{1}\orcidID{0000-0003-3313-1237} \and
  Florian~Mittelst\"adt\inst{2} \and
  Ignaz~Rutter\inst{3}\orcidID{0000-0002-3794-4406} \and
  Joachim~Spoerhase\inst{4}\orcidID{0000-0002-2601-6452} \and
  Alexander~Wolff\inst{2}\orcidID{0000-0001-5872-718X} \and
  Johannes~Zink\inst{2}\orcidID{0000-0002-7398-718X}}

\institute{%
  Institute of Theoretical Computer Science,
  Faculty of Mathematics and Computer Science,
  Jagiellonian University, Krak{\'o}w, Poland \and Universität
  Würzburg, Würzburg, Germany \and Universität Passau, Passau, Germany
  \and Max-Planck-Institut Saarbrücken, Saarbrücken, Germany}

\authorrunning{G.~Gutowski et al.}

\usepackage[ruled,linesnumbered,vlined]{algorithm2e}
\SetKwInput{KwInput}{Input}
\SetKwInput{KwOutput}{Output}
\DontPrintSemicolon

\usepackage{tikz}
\usetikzlibrary{arrows,math,patterns,shapes}

\let\leq\leqslant
\let\geq\geqslant
\let\le\leqslant
\let\ge\geqslant

\let\rho\varrho

\newcommand{\brac}[1]{{\left(#1\right)}}

\newcommand{\sbrac}[1]{{\left[#1\right]}}

\newcommand{\set}[1]{\left\{#1\right\}}

\newcommand{\norm}[1]{{\left|#1\right|}}

\newcommand{\Oh}[1]{O\brac{#1}}

\newcommand{\NP}{\ensuremath{\mathsf{NP}}\xspace}
\newcommand{\PQ}{PQ}
\newcommand{\PQP}{P}
\newcommand{\PQQ}{Q}
\newcommand{\MPQ}{MPQ}
\newcommand{\threesat}{\mbox{\textsc{3-SAT}}\xspace}
\renewcommand{\orcidID}[1]{\href{https://orcid.org/#1}{\includegraphics[scale=.03]{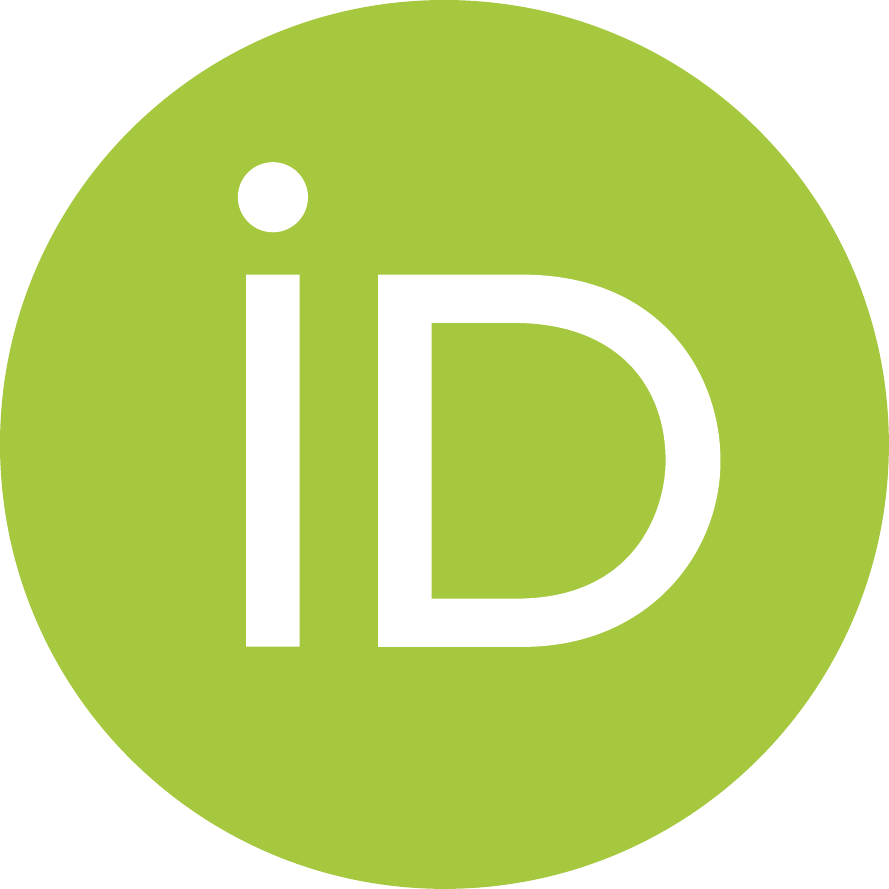}}}

\newcommand{\T}{\ensuremath{\mathcal{T}}\xspace}
\newcommand{\List}{\ensuremath{\mathcal{L}}\xspace}

\newcommand{\ie}{i.e.}
\newcommand{\eg}{e.g.}

\newcommand{\edge}[2]{\ensuremath{\{#1,#2\}}}
\newcommand{\arc}[2]{\ensuremath{(#1,#2)}}
\DeclareMathOperator{\leftc}{lc}
\DeclareMathOperator{\rightc}{rc}

\crefname{figure}{Fig.}{Figs.}
\Crefname{figure}{Figure}{Figures}
\crefname{section}{Sect.}{Sects.}
\Crefname{section}{Section}{Sections}
\crefname{appendix}{App.}{Apps.}
\Crefname{appendix}{Appendix}{Appendices}
\crefname{algocf}{Algorithm}{Algorithms}
\crefname{theorem}{Theorem}{Theorems}
\crefname{lemma}{Lemma}{Lemmas}
\crefname{conjecture}{Conjecture}{Conjectures}

\usepackage{xcolor}
\definecolor{defblue}{rgb}{0.121,0.47,0.705}
\let\emph\relax
\DeclareTextFontCommand{\emph}{\color{defblue}\em}

\graphicspath{{figures/}}

\begin{document}

\maketitle

\begin{abstract}
  A \emph{mixed graph} has a set of vertices, a set of undirected
  egdes, and a set of directed arcs.  A \emph{proper coloring} of a
  mixed graph~$G$ is a function $c$ that assigns to each vertex in~$G$
    a positive integer such that, for each edge \edge{u}{v} in~$G$,
    $c(u) \ne c(v)$ and, for each arc \arc{u}{v} in~$G$, $c(u)<c(v)$.  For
    a mixed graph~$G$, the \emph{chromatic number} $\chi(G)$ is the smallest
  number of colors in any proper coloring of~$G$.
  A \emph{directional interval graph} is a mixed graph whose vertices
  correspond to intervals on the real line.  Such a graph has an edge
  between every two intervals where one is contained in the other and an
  arc between every two overlapping intervals, directed towards the
  interval that starts and ends to the right. \par\quad
  Coloring such graphs has applications in routing edges in layered
  orthogonal graph drawing according to the Sugiyama framework; the colors correspond to the tracks for
  routing the edges.
  We show how to recognize directional interval graphs, and how to
  compute their chromatic number efficiently.
  On the other hand, for \emph{mixed interval graphs},
  \ie, graphs where two intersecting intervals can be connected by an
  edge or by an arc in either direction arbitrarily,
  we prove that computing the chromatic number is \NP-hard.

\end{abstract}

\section{Introduction}
\label{sec:intro}

A \emph{mixed graph} is a graph that contains both undirected edges
and directed arcs.  Formally, a mixed graph $G$ is a tuple~$(V,E,A)$
where~$V=V(G)$ is the set of vertices, $E=E(G)$ is the set of edges,
and~$A=A(G)$ is the set of arcs.  We require that any two vertices are
connected by at most one edge or arc.
For a mixed graph $G$, let $U(G)=(V(G),E')$ denote the
\emph{underlying undirected graph}, where $E'=E(G) \cup \{ \edge{u}{v}
\colon \arc{u}{v} \in A(G) \text{ or } \arc{v}{u} \in A(G)\}$.

A \emph{proper coloring} of a mixed graph $G$ is a function $c$
that assigns a positive integer to every vertex in $G$, satisfying
$c(u) \neq c(v)$ for every edge \edge{u}{v} in $G$, and $c(u) < c(v)$
for every arc \arc{u}{v} in $G$.  It is easy to see that a mixed graph
admits a proper coloring if and only if the arcs of~$G$ do not induce a
directed circuit.  For a mixed graph~$G$ with no directed circuit, we define the chromatic
number $\chi(G)$ as the smallest number of colors in any proper
coloring of~$G$.

The concept of mixed graphs was introduced by Sotskov and
Tanaev~\cite{SotskovT76} and reintroduced by Hansen, Kuplinsky, and de
Werra~\cite{HansenKW97} in the context of proper colorings of mixed
graphs.  Coloring of mixed graphs was used to model problems in
scheduling with precedence constraints~\cite{Sotskov20}.
It is \NP-hard in general, and it
was considered for some restricted graph classes, \eg, when the
underlying graph is a tree, a
series-parallel graph, a graph of bounded
tree-width, or a bipartite graph~\cite{FurmanczykKZ08Trees,FurmanczykKZ08SP,RiesW08}.  Mixed
graphs have also been studied in the context of (quasi-) upward planar
drawings~\cite{Binucci2016,Binucci2014,fkptw-upmpg-14}, and extensions of partial
orientations~\cite{BangJensenHZ18,KlavikKORSSV17}.

Let $\mathcal{I}$ be a set of closed non-degenerate intervals on the
real line.  The \emph{intersection graph} of $\mathcal{I}$ is the
graph with vertex set $\mathcal{I}$ where two vertices are adjacent if
the corresponding intervals intersect.  An \emph{interval graph} is a
graph~$G$ that is isomorphic to the intersection graph of some set
$\mathcal{I}$ of intervals.  We call $\mathcal{I}$ an \emph{interval
  representation} of $G$, and for a vertex $v$ in $G$, we write
$\mathcal{I}(v)$ to denote the interval that represents $v$.  A
\emph{mixed interval graph} is a mixed graph~$G$ whose underlying
graph $U(G)$ is an interval graph.

For a set $\mathcal{I}$ of closed non-degenerate intervals on the real line, the
\emph{directional intersection graph} of $\mathcal{I}$ is a mixed
graph $G$ with vertex set $\mathcal{I}$ where, for every two vertices
$u = \sbrac{l_u,r_u}$, $v = \sbrac{l_v,r_v}$ with $u$ starting to the left of $v$, \ie, $l_u \leq l_v$,
exactly one of the following conditions holds:
\begin{align*}
  \text{$u$ and $v$ are disjoint, \ie, } r_u < l_v          & \iff
    \text{$u$ and $v$ are independent in $G$,} \\
  \text{$u$ and $v$ overlap, \ie, }      l_u < l_v \leq r_u < r_v & \iff
    \text{arc \arc{u}{v} is in $G$,} \\
  \text{$u$ contains $v$, \ie, }         r_v \leq r_u       & \iff
    \text{edge \edge{u}{v} is in $G$.}
\end{align*}
A \emph{directional interval graph} is a mixed graph $G$ that is
isomorphic to the directional intersection graph of some
set~$\mathcal{I}$ of intervals.
We call $\mathcal{I}$ a \emph{directional representation} of $G$.
Similarly to interval graphs, a directional interval graph may have several different directional representations.
As there is no directed circuit in a directional interval graph $G$, $\chi(G)$ is well defined.
Observe that the endpoints in any directional representation can be perturbed so that every endpoint is unique, and the modified intervals represent the same graph.

Further, we generalize directional interval graphs and directional representations
to \emph{bidirectional interval graphs} and \emph{bidirectional representations}.
There, we assume that we have two types of intervals, which we call
\emph{left-going} and \emph{right-going}.
For left-going intervals, the edges and arcs are defined as in directional intersection graphs.
For right-going intervals, the symmetric definition applies, that is,
we have an arc $\arc{u}{v}$ if and only if $l_v < l_u \leq r_v < r_u$.
Moreover, there is an edge for every pair of a left-going and a
right-going interval that intersect.

\begin{figure}[tb]
	\centering
	\includegraphics{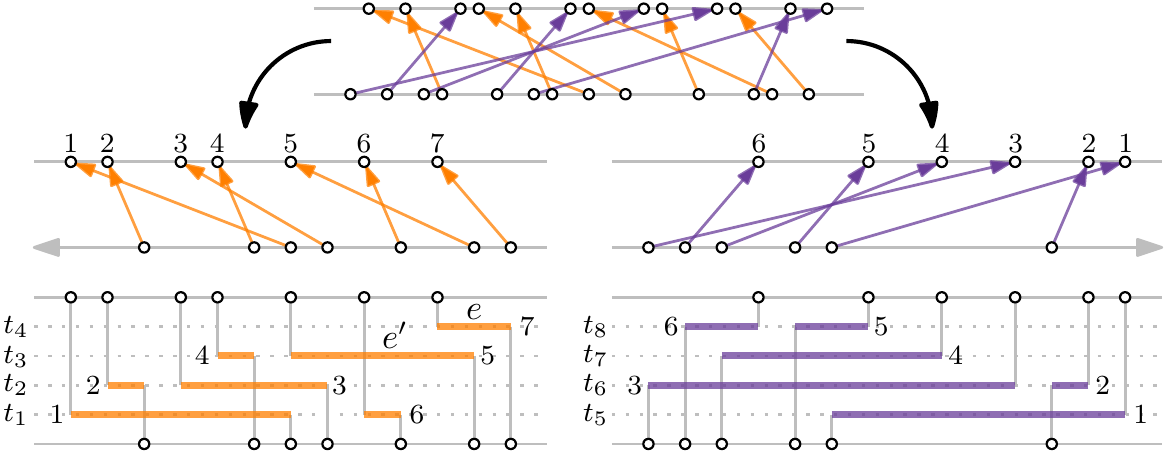}
	\caption{Separate greedy assignment of left-going and right-going
		edges to tracks.}
	\label{fig:edge-routing}
\end{figure}

Interval graphs are a classic subject of algorithmic graph theory
whose applications range from scheduling problems to analysis of
genomes~\cite{g-agtpg-80}.  Many problems that are \NP-hard for
general graphs can be solved efficiently for interval graphs.
In particular, the chromatic number of (undirected) interval
graphs~\cite{g-agtpg-80} and directed acyclic graphs~\cite{HansenKW97}
can be computed in linear time.

In this paper we combine the research directions of coloring geometric
intersection graphs and of coloring mixed graphs, by studying the
coloring of mixed interval graphs.
Our research is also motivated by the following application.

A subproblem that occurs when drawing layered graphs according to the
Sugiyama framework \cite{stt-mvuhss-TSMC81} is the edge routing step.
This step is applied to every pair of consecutive layers.  Zink et
al.~\cite{zwbw-ldugg-CGTA22} formalize this
for orthogonal edges as follows.  Given a set
of points on two horizontal lines (corresponding to the vertices on
two consecutive layers) and a perfect matching between the points on
the lower and those on the upper line, connect the matched pairs of
points by x- and y-monotone rectilinear paths.  Since we can assume
that no two points have the same x-coordinate, each pair of
points can be connected by a path that consists of three axis-aligned
line segments; a vertical, a horizontal, and another vertical one; see
\cref{fig:edge-routing}.  We refer to the interval that corresponds to
the vertical projection of an edge to the x-axis as the \emph{span} of
that edge.
We direct all edges upward.  This allows us to classify the edges into
\emph{left-} vs.\ \emph{right-going}.

Now the task is to map the horizontal pieces to horizontal ``tracks''
between the two layers such that no two such pieces overlap and no two
edges cross twice.  This implies that any two edges whose spans
intersect must be mapped to different tracks.  If there is a
left-going edge~$e$ whose span overlaps that of another left-going
edge~$e'$ that lies further to the left (see
\cref{fig:edge-routing}), then $e$ must be mapped to a
higher track than~$e'$ to avoid crossings.  The symmetric statement
holds for pairs of right-going edges.  The aim is to minimize the
number of tracks in order to get a compact layered drawing of the
original graph.
This corresponds to minimizing the number of colors in a proper
coloring of a bidirectional interval graph.
Zink et al.\ solve this combinatorial problem heuristically.
They greedily construct two colorings (of left-going edges and of
right-going edges) and combine the colorings by assigning separate
tracks to the two directions; see \cref{fig:edge-routing}.

\paragraph{Our contribution.}

We first show that the above-mentioned greedy algorithm of Zink et
al.~\cite{zwbw-ldugg-CGTA22} colors directional interval graphs with
the minimum number of colors; see \cref{sec:greedy}.  This yields a
simple 2-approximation algorithm for the bidirectional case.
Then, we prove that computing the chromatic number of a mixed interval
graph is \NP-hard; see \cref{sec:hardness}.  This result extends to
proper interval graphs; see \cref{sec:proper}.
Finally, we present an efficient algorithm that recognizes directional
interval graphs; see \cref{sec:recognition}.  Our algorithm is based
on \PQ-trees and the recognition of two-dimensional posets.  It can
construct a directional interval representation of a yes-instance in
quadratic time.

We postpone the proofs of statements with a (clickable)
``$\star$'' to the appendix.

\section{Coloring Directional Interval Graphs}
\label{sec:greedy}
We prove that the greedy algorithm of Zink et
al.~\cite{zwbw-ldugg-CGTA22} computes an optimal coloring for a given directional interval representation of $G$.
If we are not given a representation (i.e., a set of intervals)
but only the graph, we obtain
a representation in quadratic time by \cref{thm:recognition}.
The greedy algorithm proceeds analogously to the classic greedy coloring algorithm for (undirected) interval graphs.
Also our optimality proof follows, on a high level, the strategy of relating the coloring to a large clique.
In our setting, however, the underlying geometry is more intricate,
which makes the optimality proof as well as a fast implementation more
involved. 
The algorithm works as follows; see \cref{fig:edge-routing} (left) for an example.
\smallskip

\noindent
\fbox{%
	\begin{minipage}{.97\linewidth}
		\textsc{Greedy Algorithm.}
		Iterate over the given intervals in increasing order of
		their left endpoints.  For each interval~$v$, assign $v$ the
		smallest available color~$c(v)$.  A~color~$k$ is \emph{available}
		for $v$ if, for any interval~$u$ that has already been colored,
		$k \ne c(u)$ if~$u$ contains~$v$ and $k>c(u)$ if~$u$ overlaps~$v$.
	\end{minipage}
}

\medskip

A naive implementation of the greedy algorithm runs in quadratic time. Using augmented binary search trees, we can speed it up to optimal $O(n \log n)$ time.

\begin{restatable}[{\hyperref[lem:runtimegreedy*]{$\star$}}]{lemma}{RuntimeGreedy}
	\label{lem:runtimegreedy}
	The greedy algorithm can be implemented to color $n$
	intervals in $\Oh{n \log n}$ time, which is
	optimal assuming the comparison-based model.
\end{restatable}

Next we show that the greedy algorithm computes an optimal proper
coloring.  This also yields a simple 2-approximation for the
bidirectional case.

\begin{theorem}
	\label{clm:greedy-is-optimal}
	Given a directional representation of a directional interval
	graph~$G$, the greedy algorithm computes a proper coloring of~$G$
	with~$\chi(G)$ many colors.
\end{theorem}

\begin{proof}
	The \emph{transitive closure} $G^+$ of $G$ is the graph that we obtain
	by exhaustively adding transitive arcs, i.e., if there are arcs
	$\arc{u}{v}$ and $\arc{v}{w}$, we add the arc $\arc{u}{w}$ if
	absent. Clearly, no pair of adjacent intervals in
	the underlying undirected graph~$U(G^+)$ of $G^+$ can have
	the same color in a proper coloring of~$G$. Therefore,
	$\omega(U(G^+))\leq \chi(G)$ where $\omega(U(G^+))$ denotes the size
	of a largest clique in $U(G^+)$. We show below that the greedy algorithm
	computes a coloring with at most $\omega(U(G^+))$ many colors, which
	must therefore be optimal. For $v\in V$ let
	$\mathcal{I}_\textsf{in}(v)$ be the set of intervals having an arc to~$v$
	in $G$.
	
	Let $c$ be the coloring computed by our greedy coloring algorithm.
	Since we always pick an available color, $c$ is a proper
	coloring. To prove optimality of~$c$, we show the existence of a
	clique in $U(G^+)$ of cardinality $c_{\max}=\max_{v\in V}c(v)$.
	
	\begin{figure}[t]
		\centering \includegraphics[page=1]{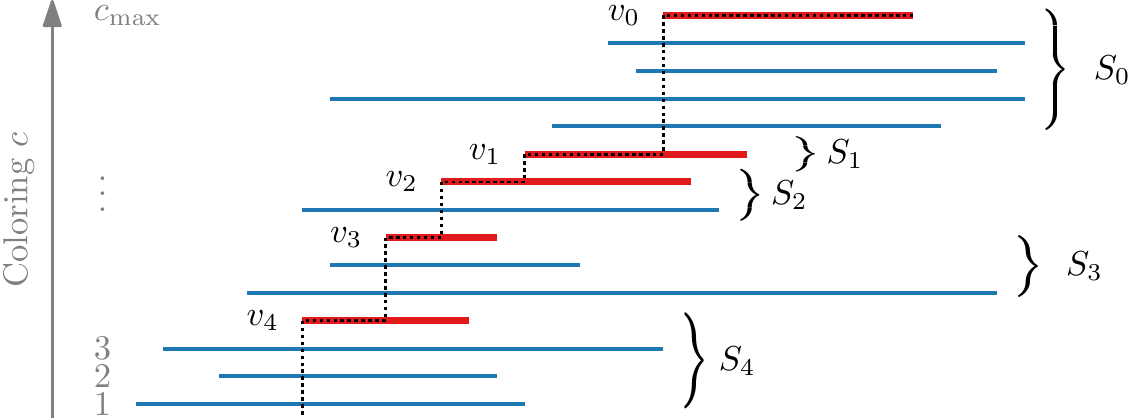}
		\caption{A staircase and its intermediate intervals, which form a
			clique in $U(G^+)$.}
		\label{fig:staircase}
	\end{figure}
	
	Consider an interval $v_0 = [l_0, r_0]$ of color~$c_{\max}$. Among
	$\mathcal{I}_\textsf{in}(v_0)$, let~$v_1$ be the unique interval
	with the largest color (all intervals in
	$\mathcal{I}_\textsf{in}(v_0)$ have different colors as they
	share the point $l_0$).  We call $v_1$ the \emph{step below $v_0$}.
	We repeat this argument to find the step $v_2$ below $v_1$ and
	so on.  For some $t \ge 0$, there is a $v_t$ without a step below
	it, namely where $\mathcal{I}_\textsf{in}(v_t) = \emptyset$.  We
	call the sequence $v_0, v_1, \dots, v_t$ a \emph{staircase} and each
	of its intervals a \emph{step}; see~\cref{fig:staircase}.  Clearly,
	$\arc{v_j}{v_i}$ is an arc of~$G^+$ for~$0 \le i < j \le t$.  In
	particular, the staircase is a clique of size $t+1$ in $U(G^+)$.
	Next we argue about the intervals with colors in-between the steps.
	
	For a step~$v_i = [l_i,r_i]$, $i \in \set{0,\dots,t}$, let~$S_i$
	denote the set of intervals that contain the point~$l_i$ and have a
	color~$x \in \set{c(v_{i+1})+1,c(v_{i+1})+2,\dots,c(v_i)}$; see
	\cref{fig:staircase}.  Note that~$v_i \in S_i$ and, by the
	definition of steps, each interval in~$S_i$ contains~$v_i$.
	Observe that $|S_i| = c(v_{i}) - c(v_{i+1})$, as
	otherwise the greedy algorithm would have assigned a smaller color to~$v_i$.  It
	follows that~$c_{\max} = \sum_{i=0}^t |S_i|$.
	
	We claim that~$S = \bigcup_{i=0}^t S_i$ is a clique in~$U(G^+)$.
	Let~$u \in S_i$, $v \in S_l$ such that~$u \cap v = \emptyset$
	(otherwise they are clearly adjacent in~$U(G^+)$). Assume without
	loss of generality that~$i<l$.  Let~$j,k$ be the largest and
	smallest index so that~$v_j \cap u \ne \emptyset$ and
	$v_k \cap v \ne \emptyset$, respectively.  Observe
	that~$u \cap v = \emptyset$, $u \cap v_{i+1} \ne \emptyset$,
	and~$v \cap v_{l-1} \ne \emptyset$ imply~$i < j < l$ and
	$i < k < l$.  Since~$u$ does not intersect~$v_{j+1}$, it overlaps
	with~$v_j$, i.e., $G$ contains the arc~$\arc{v_j}{u}$ and likewise,
	since~$v$ does not intersect $v_{k-1}$, it overlaps with~$v_k$,
	i.e., $G$ contains the arc~$\arc{v}{v_k}$.
	
	If $j<k$, then $G^+$ contains $\arc{v}{v_k}$ and~$\arc{v_k}{v_j}$,
	and therefore $\arc{v}{v_j}$.  If $j\geq k$, then~$v_j$ is adjacent
	to both~$u$ and~$v$, and since~$u,v$ are disjoint, $v_j$ overlaps with~$u$ and~$v$, i.e., $G$ contains $\arc{v}{v_j}$.  In either
	case, the presence of \arc{v}{v_j} and \arc{v_j}{u} implies that
	$G^+$ contains \arc{v}{u}.  It follows that~$S$ forms a clique
	in~$U(G^+)$.
\end{proof}

\begin{restatable}[{\hyperref[cor:approx*]{$\star$}}]{corollary}{Approx}
	\label{cor:approx}
	There is an $\Oh{n \log n}$-time algorithm that, given a bidirectional interval representation, computes a 2-approximation of an optimal proper coloring of the corresponding bidirectional interval graph.
\end{restatable}

\section{Coloring Mixed Interval Graphs}
\label{sec:hardness}
In this section, we show that computing the chromatic number of a mixed interval graph is \NP-hard.
Recall that the chromatic number can be computed efficiently for
interval graphs~\cite{g-agtpg-80}, directed acyclic
graphs~\cite{HansenKW97}, and directional interval graphs
(\cref{clm:greedy-is-optimal}).  In other words, coloring interval
graphs becomes \NP-hard only if edges and arcs are combined in a
non-directional way.

\begin{theorem}\label{thm:hardness}
	Given a mixed interval graph $G$ and a number $k$,
	it is \NP-complete to decide whether $G$ admits a proper coloring
	with at most $k$ colors.
\end{theorem}

\begin{proof}
	Containment in \NP is clear since a specific coloring with $k$ colors
	serves as a certificate of polynomial size.  We prove \NP-hardness by
	a polynomial-time reduction from \threesat.  The high-level idea is as
	follows.  We are given a \threesat formula $\Phi$ with variables
	$v_1, v_2, \dots, v_n$, and clauses $c_1, c_2, \dots, c_m$, where each
	clause contains at most three literals.  A literal is a variable or a
	negated variable~-- we refer to them as a \emph{positive} or a
	\emph{negative} occurrence of that variable.  From $\Phi$, we
	construct in polynomial time a mixed interval graph~$G_\Phi$ with the
	property that $\Phi$ is satisfiable if and only if $G_\Phi$ admits a
	proper coloring with $6n$ colors.
	
	To prove that~$G_\Phi$ is a mixed interval graph, we present an interval representation of~$U(G_\Phi)$ and specify which pairs of intersecting intervals are connected by a directed arc, assuming that all other pairs of intersecting intervals are connected by an edge.
	The graph $G_\Phi$ has the property that the color of many of the intervals is fixed in every proper coloring with $6n$ colors.
	In our figures, the x-dimension corresponds to the real line that contains the interval, whereas we indicate its color
	by its position in the~y-dimension~-- thus, we also refer to a color as a \emph{layer}.
	In this model, our reduction has the property that~$\Phi$ is satisfiable if and only if the intervals of $G_\Phi$ admit a drawing that fits into $6n$ layers.
	
	Our construction consists of a \emph{frame} and $n$ \emph{variable gadgets} and $m$ \emph{clause gadgets}.
	Each variable gadget is contained in a horizontal strip of height~$6$ that spans the whole construction,
	and each clause gadget is contained in a vertical strip of width~$4$ and height~$6n$.
	The strips of the variable gadgets are pairwise disjoint,
	and likewise the strips of the clause gadgets are pairwise disjoint.
	
	\paragraph{Frame.}
	
	See \cref{fig:frame-structure}.
	The frame consists of six intervals~$f_i^1, f_i^2,\dots,f_i^6$ for each of the variables~$v_i$, $i=1,\dots, n$.
	All of these intervals start at position~$0$ and extend from the left into the construction.
	The intervals~$f_i^2,f_i^4,f_i^6$ end at position~$1$.
	The intervals~$f_i^1$ and~$f_i^5$ extend to the very right of the construction.
	Interval~$f_i^3$ ends at position~$3$.
	Further, there are arcs~$\arc{f_i^j}{f_i^{j+1}}$ for~$j=1,\dots,5$ and~$\arc{f_i^6}{f_{i+1}^1}$ for~$i=1,\dots,n-1$.
	This structure guarantees that any proper coloring with colors $\set{1,2,\ldots,6n}$ assigns color $6(i-1)+j$ to interval $f_i^j$.
	
	\begin{figure}[t]
		\begin{subfigure}{.45 \linewidth}
			\centering
			\includegraphics[page=1]{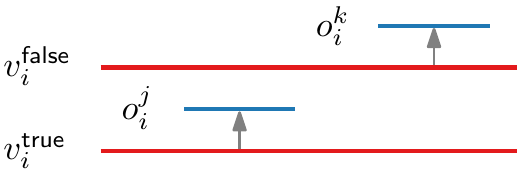}
			\caption{$v_i$ is \textsf{true}
					and appears positively in~$c_j$, and negatively
					in~$c_k$.}
			\label{fig:variable-gadget-true}
		\end{subfigure}
		\hfill
		\begin{subfigure}{.45 \linewidth}
			\centering
			\includegraphics[page=2]{variable-gadget}
			\caption{$v_i$ is \textsf{false}
					and appears positively in~$c_j$, and negatively
					in~$c_k$.}
			\label{fig:variable-gadget-false}
		\end{subfigure}
		
		\smallskip
		
		\begin{subfigure}{1.0 \linewidth}
			\centering
			\includegraphics[page=1]{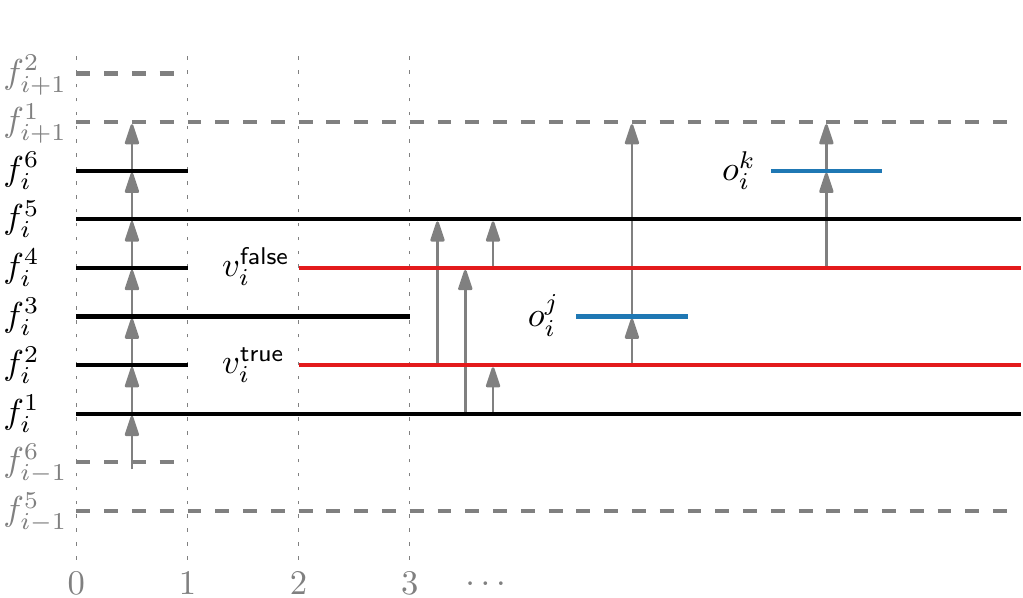}
			\caption{Frame.}
			\label{fig:frame-structure}
		\end{subfigure}
		\caption{A variable gadget for a variable $v_i$.}
		\label{fig:variable-gadget}
	\end{figure}
	
	\paragraph{Variable Gadget.}
	
	See \cref{fig:variable-gadget-true,fig:variable-gadget-false}.
	For each variable~$v_i$, $i = 1, \dots, n$, we have two intervals $v_i^\mathsf{false}$ and $v_i^\mathsf{true}$, which start at position~2 and extend to the very right of the construction.
	Moreover, they both have an incoming arc from~$f_i^1$ and an outgoing arc to~$f_i^5$.
	This guarantees that they are drawn in the layers of~$f_i^2$ and~$f_i^4$, however their ordering can be chosen freely.
	We say that $v_i$ is set to \textsf{true} if $v_i^{\mathsf{true}}$ is below $v_i^{\mathsf{false}}$, and $v_i$ is set to \textsf{false} otherwise.
	
	For each occurrence of $v_i$ in a clause $c_j$, $j = 1, \dots, m$, we create an interval~$o_i^j$ within the clause gadget of~$c_j$.
	There is an arc $\arc{v_i^\mathsf{true}}{o_i^j}$ for a positive occurrence and an arc $\arc{v_i^\mathsf{false}}{o_i^j}$ for a negative occurrence as well as an arc~$\arc{o_i^j}{f_{i+1}^1}$ if $i<n$.
	This structure guarantees that $o_i^j$ is drawn either in the same layer as~$f_i^3$ or as~$f_i^6$.
	However, drawing $o_i^j$ in the layer of~$f_i^3$
	(which lies between $v_i^\mathsf{true}$ and~$v_i^\mathsf{false}$) is possible if and only if
	the chosen truth assignment of $v_i$ satisfies~$c_j$.
	
	\paragraph{Clause Gadget.}
	
	\begin{figure}[t]
		\begin{subfigure}{.45 \linewidth}
			\centering
			\includegraphics[page=3]{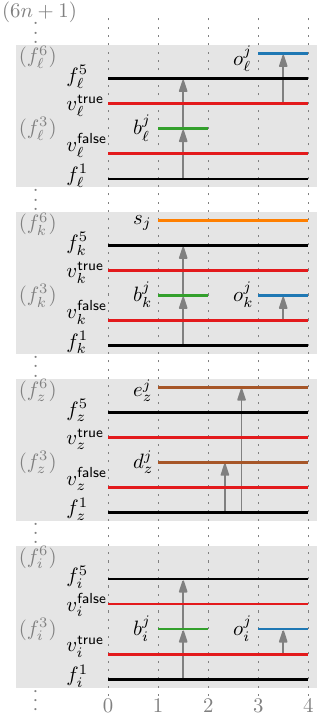}
			\caption{$c_j$ is satisfied.}
			\label{fig:clause-gadget-satisfied}
		\end{subfigure}
		\hfill
		\begin{subfigure}{.45 \linewidth}
			\centering
			\includegraphics[page=2]{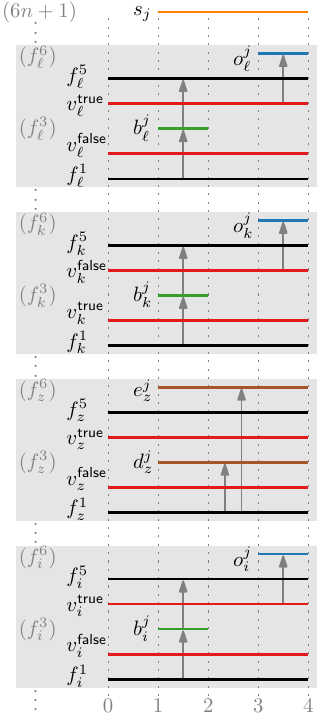}
			\caption{$c_j$ is not satisfied.}
			\label{fig:clause-gadget-not-satisfied}
		\end{subfigure}
		\caption{A clause gadget for a clause $c_j = v_i \lor \lnot
			v_k \lor v_\ell$, where $z \notin \set{i,k,\ell}$.}
		\label{fig:clause-gadget}
	\end{figure}
	
	See \cref{fig:clause-gadget}.
	Our clause gadget starts at position~$4j$,
	relative to which we describe the following positions.
	Consider a fixed clause~$c_j$ that contains variables~$v_i,v_k,v_\ell$.
	We create an interval~$s_j$ of length $3$ starting at position~$1$.
	The key idea is that~$s_j$ can be drawn in the layer of~$f_i^6, f_k^6$ or~$f_\ell^6$, but only if $o_i^j$, $o_k^j$ or~$o_\ell^j$,
	each of which has length~$1$ and starts at position~$3$,
	is not drawn there.
	This is possible iff
	the corresponding variable satisfies the clause.
	
	To ensure that~$s_j$ does not occupy any other layer, we block all the other layers.
	More precisely, for each variable~$v_z$ with $z \notin \set{i,k,\ell}$, we create \emph{dummy intervals}~$d_z^j,e_z^j$ of length~$3$ starting at
	position~$1$ that have arcs from~$f_z^1$ and to~$f_{z+1}^1$.
	These arcs force
	$d_z^j, e_z^j$ to be drawn in the layers of~$f_z^3$ and~$f_z^6$, thereby ensuring that~$s_j$ is not placed in any layer associated with the variable~$z$.
	
	Similarly, for each $z \in \set{i,k,\ell}$, we create a blocker~$b_z^j$
	of length~$1$ starting at position~$1$
	that has arcs from~$f_z^1$ and to~$f_z^5$.
	This fixes $b_z^j$ to the layer of~$f_z^3$ (since the layers of~$f_z^2$ and~$f_z^4$
	are occupied by~$v_z^\mathsf{true}$ and $v_z^\mathsf{false}$), thereby ensuring that,
	among all layers associated with~$v_z$, $s_j$ can only be drawn in the layer of~$f_z^6$.
	
	\paragraph{Correctness.}
	
	Consider for each clause $c_j$ with variables $v_i$, $v_k$, and
	$v_\ell$ the corresponding clause gadget.  To achieve a total height
	of at most $6n$, $s_j$ needs to be drawn in the same layer as some interval
	of the frame.  Due to the presence of the dummy intervals,
	the only available layers are the ones of~$f_z^6$ for~$z \in \{i,k,\ell\}$.
	However, the layer of~$f_z^6$ is only free if~$o_z^j$ is not there, which is the
	case if and only if~$o_z^j$ is drawn in the layer of~$f_z^3$.  By
	construction, this is possible if and only if the variable~$v_z$ is in
	the state that satisfies clause~$j$.
	Otherwise we need an extra $(6n+1)$-th layer.
	Both situations are
	illustrated in \cref{fig:clause-gadget}.  Hence, $6n$ layers are
	sufficient if and only if the variable gadgets represent a truth
	assignment that satisfies all the clauses of $\Phi$.  The mixed interval
	graph $G_\Phi$ has polynomial size and can be constructed in
	polynomial time.
\end{proof}

A \emph{proper interval graph} is an interval graph that
admits an interval representation of the underlying graph in which
none of the intervals properly contains another interval.
We can slightly adjust the reduction presented in the proof of \cref{thm:hardness} to make $G_\Phi$ a \textit{mixed} proper interval graph.

\begin{restatable}[{\hyperref[clm:hardnesspropper*]{$\star$}}]{corollary}{HardnessPropper}
	\label{clm:hardnesspropper}
	Given a mixed proper interval graph $G$ and a number $k$,
	it is \NP-complete to decide whether $G$ admits a proper coloring
	with at most~$k$~colors.
\end{restatable}

\section{Recognizing Directional Interval Graphs}
\label{sec:recognition}
In this section we present a recognition algorithm for directional interval graphs.
Given a mixed graph~$G$, our algorithm decides whether $G$ is a
directional interval graph, and additionally if the answer is yes, it
constructs a set of intervals representing~$G$.
The algorithm works in two phases.
The first phase carefully selects a rotation of the \PQ-tree of~$U(G)$.
This fixes the order of maximal cliques in the interval representation of~$U(G)$.
In the second phase, the endpoints of the intervals are perturbed so that the edges and arcs in~$G$ are represented correctly.
This is achieved by checking that an auxiliary poset is two-dimensional.

\PQ-trees of interval graphs~\cite{LuekerB79} and realizers
of two-dimensional posets~\cite{McConnellS99} can be constructed in
linear time.  Our algorithm runs in quadratic time, but we suspect
that a more involved implementation can achieve linear running time.

For a set of pairwise intersecting intervals on the real line, let the \emph{clique point} be the leftmost point on the real line that lies in all the intervals.
Given an interval representation of an interval graph~$G$, we get a linear order of the maximal cliques of~$G$ by their clique points from left to right.
Booth and Lueker~\cite{LuekerB79} showed that a graph~$G$ is an interval graph if and only if the maximal cliques of~$G$ admit a \emph{consecutive arrangement}, \ie, a linear order such that, for each vertex~$v$, all the maximal cliques containing~$v$ occur consecutively in the order.
They have also introduced a data structure called \PQ-tree
that encodes all possible consecutive arrangements of~$G$.
We present our algorithm in terms of modified \PQ-trees
(\MPQ-trees, for short) as described by Korte and
M{\"o}hring~\cite{KorteM85,KorteM89}.
We briefly describe \MPQ-trees in the next few paragraphs;
see~\cite{KorteM89} for a proper introduction.

An \emph{\MPQ-tree}~$T$ of an interval graph~$G$ is a rooted, ordered tree with two types of nodes: \PQP-nodes and \PQQ-nodes, joined by links.
Each node can have any number of children and a set of consecutive links joining a \PQQ-node~$x$ with children is called a \emph{segment} of~$x$.
Further, each vertex $v$ in~$G$ is assigned either to one of the \PQP-nodes, or to a segment of some \PQQ-node.
Based on this assignment, we \emph{store} $v$ in the links of~$T$.
If~$v$ is assigned to a \PQP-node~$x$, we store $v$ in the link just above $x$ in $T$ (adding a dummy link above the root of $T$).
If $v$ is assigned to a segment of a \PQQ-node~$x$, we store $v$ in each link of the segment.
For a link \edge{x}{y}, let $S_{xy}$ denote the set of vertices stored in \edge{x}{y}.
We say that $v$ is \emph{above} (\emph{below}, resp.) a node~$x$ if $v$ is stored in any of the links on the upward path (in any of the links on some downward path, resp.) from $x$ in $T$.
We write $A^T_x$ ($B^T_x$, resp.) for the set of all vertices in~$G$ that are above (below, resp.) node~$x$.

The \emph{frontier} of~$T$ is the sequence of the sets $A^T_x$, where
$x$ goes through all leaves in~$T$ in the order of~$T$.
Given an \MPQ-tree~$T$, one can obtain another \MPQ-tree, which is called a \emph{rotation} of~$T$,
by arbitrarily permuting the order of the children of \PQP-nodes and
by reversing the orders of the children of some \PQQ-nodes.  The defining property
of the \MPQ-tree~$T$ of a graph $G$ is that each leaf $x$ of~$T$ corresponds to a maximal clique $A^T_x$ of~$G$ and the frontiers of rotations of~$T$ correspond bijectively to the consecutive arrangements of~$G$.
Observe that any two vertices adjacent in~$G$ are stored in links that are connected by an upward path in~$T$.
We say that~$T$ \emph{agrees} with an interval representation $\mathcal{I}$ of~$G$ if the order of the maximal cliques of~$G$ given by their clique points in $\mathcal{I}$ from left to right is the same as in the frontier of~$T$.
We assume the following properties of the \MPQ-tree (see \cite{KorteM89}, Lemma~2.2):
\begin{itemize}
	\item For a \PQP-node~$x$ with children $y_1,\ldots,y_k$, for every $i=1,\ldots,k$,
	there is at least one vertex stored in link \edge{x}{y_i} or below $y_i$, \ie, $S_{xy_i} \cup B^T_{y_i} \neq \emptyset$.
	\item For a \PQQ-node~$x$ with children $y_1,\ldots,y_k$, we have $k \geq 3$.
	Further, for $S_i=S_{xy_i}$, we have:
	\begin{itemize}
		\item $S_1 \cap S_k = \emptyset$,
		$B^T_{y_1} \neq \emptyset$, $B^T_{y_k} \neq \emptyset$,
		$S_{1} \subsetneq S_{2}$, $S_{k} \subsetneq S_{k-1}$,
		\item $(S_i \cap S_{i+1})\setminus S_1 \neq \emptyset$, $(S_{i-1} \cap S_{i})\setminus S_k \neq \emptyset$, for $i=2,\ldots,k-1$.
	\end{itemize}
\end{itemize}

A \emph{partially ordered set}, or a \emph{poset} for short, is a transitive directed acyclic graph.
A poset~$P$ is \emph{total} if, for every pair of vertices~$u$
and~$v$, there is either an arc \arc{u}{v} or an arc \arc{v}{u}
in~$P$.
We can conveniently represent a total poset~$P$ by a linear order of its vertices $v_1 < v_2 < \dots < v_n$ meaning that there is an arc \arc{v_i}{v_j} for each $1 \leq i < j \leq n$.
A poset~$P$ is \emph{two-dimensional} if the arc set of~$P$ is the
intersection of the arc sets of two total posets on the same set
of vertices as~$P$.
McConnell and Spinrad~\cite{McConnellS99} gave a linear-time algorithm
that, given a directed graph~$D$ as input, decides whether~$D$ is a
two-dimensional poset.  If the answer is yes, the algorithm also
constructs a \emph{realizer}, that is, (in this case) two linear
orders $(R_1,R_2)$ on the vertex set of~$D$ such that
\begin{align*}
\textrm{arc \arc{u}{v} is in~$D$} & \iff
\sbrac{\brac{\textrm{$u < v$ in $R_1$}} \wedge \brac{\textrm{$u < v$ in $R_2$}}} \textrm{.}
\end{align*}

The main result of this section is the following theorem.
\begin{restatable}[{\hyperref[thm:recognition*]{$\star$}}]{theorem}{ThmRecognition}
	\label{thm:recognition}
	There is an algorithm that, given a mixed graph $G$, decides
	whether $G$ is a directional interval graph.  The algorithm runs in
	$\Oh{|V(G)|^2}$ time
	and produces a directional representation of~$G$ if $G$ admits one.
\end{restatable}

The algorithm runs in two phases that we introduce in separate lemmas.

\begin{lemma}[Rotating \PQ-trees]\label{lem:recognition_rotation}
	There is an algorithm that, given a directional interval graph~$G$,
	constructs an \MPQ-tree~$T$ that agrees with some
	directional representation of~$G$.
\end{lemma}

\begin{proof}
	Given a mixed graph $G$, if~$G$ is a directional interval graph,
	then clearly $U(G)$ is an interval graph and we can construct an
	\MPQ-tree~$T$ of~$U(G)$ in linear time using the algorithm by Korte
	and M{\"o}hring~\cite{KorteM89}.  We call a rotation of~$T$ 
	\emph{directional} if it agrees with some directional representation
	of~$G$.  As we assume~$G$ to be a directional interval graph, there
	is at least one directional rotation~$\tilde T$ of~$T$, and our goal
	is to find some directional rotation of~$T$.  Our algorithm decides
	the rotation of each node in~$T$ independently.
	
	\paragraph{Rotating \PQQ-nodes.}
	Let $y_1,\ldots,y_k$ be the children of a \PQQ-node~$x$ in~$T$.
	We are to decide whether to reverse the order of the children of~$x$.
	Let $S_i=S_{xy_i}$,
	let~$\ell = \max\set{i \colon S_1 \cap S_i \ne \emptyset}$, and let
	$u \in S_1 \cap S_\ell$.  We have $\ell < k$, and there is some vertex
	$v \in (S_\ell \cap S_{\ell+1}) \setminus S_1$.  This implies that
	$u$ and $v$ are assigned to overlapping segments of~$x$.
	Thus, the intervals representing~$u$ and~$v$ overlap in every interval
	representation of~$U(G)$.
	Hence, $u$ and~$v$ are connected by an arc
	in~$G$, and the direction of this arc determines the only possible
	rotation of~$x$ in any directional rotation of~$T$, \eg{},
	if $\arc{u}{v}$ is an arc in $G$ and the segment of~$u$ is to the
	right of the segment of~$v$, then reverse the order of the children
	of~$x$.

	\paragraph{Rotating \PQP-nodes.}
	Let $y_1,\ldots,y_k$ be the children of a \PQP-node~$x$ in~$T$.
	For each $i=1,\ldots,k$, let $B_i=S_{xy_i} \cup B^T_{y_i}$, and let
	$B=\bigcup_{i=1}^kB_i$.  The properties of the
	\MPQ-tree give us that (i)~every vertex in~$A^T_x$ is adjacent
	in~$U(G)$ to every vertex in~$B$, (ii)~none of
	the $B_i$ is empty, and (iii)~for any two vertices $b_i \in B_i$,
	$b_j \in B_j$ with $i\neq j$, we have that~$b_i$ and~$b_j$ are
	independent in~$G$.
	
	Assume that there is an arc \arc{b_i}{a} directed from some
	$b_i \in B_i$ to some $a \in A^T_x$.  We claim that any
	rotation~$T'$ of~$T$ that does not put~$y_i$ as the first child
	of~$x$ is not directional.  Assume the contrary.  Let~$y_j$,
	$j \neq i$ be the first child of~$x$ in~$T'$, let $\mathcal{I}$ be a
	directional representation that agrees with~$T'$, and let~$b_j$
	be some vertex in~$B_j$.
	The left endpoint of $\mathcal{I}(a)$ is to the
	right of the left endpoint of $\mathcal{I}(b_i)$ as \arc{b_i}{a} is
	an arc.  The right endpoint of $\mathcal{I}(b_j)$ is to the left of
	the left endpoint of $\mathcal{I}(b_i)$ as~$T'$ puts~$y_j$
	before~$y_i$.  Thus, $\mathcal{I}(b_j)$ and $\mathcal{I}(a)$ are
	disjoint, a~contradiction.
	
	Similarly, there are directed arcs from~$A^T_x$ to at
	most one set of type~$B_i$.  If there are any, the corresponding
	child $y_i$ is in the last
	position in every directional rotation of~$T$.  Our algorithm
	rotates the child~$y_i$ ($y_j$) with an arc from~$B_i$ to~$A^T_x$
	(from~$A^T_x$ to~$B_j$) to the first (last) position, should such
	children exist, and leaves the other children as they are in~$T$.
	It remains to show that the resulting rotation of~$T$ is directional; see~\cref{sec:rotation}.
\end{proof}

\begin{lemma}[Perturbing Endpoints]\label{lem:recognition_dimension}
	There is an algorithm that, given an \MPQ-tree~$T$ that agrees with
	some directional representation of a graph~$G$, constructs a
	directional representation $\mathcal{I}$ of~$G$ such that $T$ agrees
	with $\mathcal{I}$.
\end{lemma}
\begin{proof}
	The frontier of~$T$ yields a fixed order of maximal cliques $C_1,\ldots,C_k$ of~$G$.
	Given this order, we construct the following auxiliary poset~$D$.
	First, we add two independent chains of length~$k+1$ each: vertices $a_1,\ldots,a_{k+1}$ with arcs \arc{a_i}{a_j} for $1\leq i<j\leq k+1$, and vertices $b_1,\ldots,b_{k+1}$ with arcs \arc{b_i}{b_j} for $1\leq i<j\leq k+1$.
	Then, for each vertex~$v$ in~$G$, let $\leftc(v)$ and $\rightc(v)$
	denote the indices of the leftmost and of the rightmost clique in
	which~$v$ is present, respectively.
	Now we add to~$D$ vertex~$v$ plus, for $1\leq i \leq \leftc(v)$, the
	arc \arc{a_i}{v} and, for $1\leq i \leq \rightc(v)$, the arc \arc{b_i}{v}.
	Further, for each arc \arc{u}{v} in~$G$, we add \arc{u}{v} to~$D$.
	Lastly, for any two vertices~$u$ and $v$ that are independent in~$G$
	and that fulfill $\rightc(u) < \leftc(v)$, we add an arc \arc{u}{v} to~$D$.
	We claim that~$G$ is a directional interval graph if and only if~$D$
	is a two-dimensional poset.
	
	First assume that~$G$ is a directional interval graph and fix a directional interval representation of~$G$ whose intervals all have distinct endpoints.
	For $i=1,\ldots,k$, let~$L_i$ be the sequence of all the vertices~$v$
	in~$G$ for which $\leftc(v)=i$, in the order of their left endpoints.
	Similarly, let~$R_i$ be the sequence of all the vertices~$v$ in~$G$
	for which $\rightc(v)=i$, in the order of their right endpoints.
	The following two linear orders $L$ and $R$ of the vertices of~$D$
	yield a realizer of~$D$:
	\begin{align*}
	L &= b_1 < b_2 < \ldots < b_k < a_1 < L_1 < a_2 < L_2 < \ldots < a_k < L_k < a_{k+1}\textrm{,} \\
	R &= a_1 < a_2 < \ldots < a_k < b_1 < R_1 < b_2 < L_2 < \ldots < b_k < R_k < b_{k+1}\textrm{.}
	\end{align*}
	
	Now, for the other direction, assume that we have a two-dimensional realizer of~$D$.
	As $b_{k+1}$ and~$a_1$ are independent in~$D$, we have that $b_{k+1} < a_1$ in exactly one of the orders in the realizer.
	We call this order~$L$, and the other one~$R$.
	As $a_{k+1}$ and~$b_1$ are independent in~$D$ and $b_1 < b_{k+1} < a_1 < a_{k+1}$ in~$L$, we have that $a_{k+1} < b_1$ in~$R$.
	For each $i=1,\ldots,k$, define~$L_i$ as the sequence of vertices in~$G$ appearing between $a_{i}$ and $a_{i+1}$ in the order~$L$.
	Similarly, let~$R_i$ be the sequence of vertices in~$G$ appearing between $b_{i}$ and $b_{i+1}$ in the order~$R$.
	Observe that, for every vertex~$v$, we have that $a_{\leftc(v)} < v$
	in~$D$ and that $a_{\leftc(v)+1}$ and~$v$ are independent in~$D$.
	As $a_{\leftc(v)+1} \leq a_{k+1} < b_1 \leq b_{\rightc(v)} < v$ in~$R$, we have $v < a_{\leftc(v)+1}$ in~$L$.
	Thus, $v$~is in $L_{\leftc(v)}$ and, by a similar argument, $v$~is in $R_{\rightc(v)}$.
	
	Now we are ready to construct a directional interval representation $\mathcal{I}$ of~$G$.
	For each $i=1,\ldots,k$, we select $\norm{L_i}$ different real points
	in $(i-\frac{1}{2},i)$ and $\norm{R_i}$ different real points in
	$(i,i+\frac{1}{2})$.
	For a vertex~$v$ that appears on the~$i$-th position in~$L_{\leftc(v)}$
	and on the~$j$-th position in $R_{\rightc(v)}$, we choose the~$i$-th
	point in $(\leftc(v)-\frac{1}{2},\leftc(v))$ as the
	left endpoint, and the~$j$-th point in
	$(\rightc(v),\rightc(v)+\frac{1}{2})$ as the right endpoint.
	Such a set of intervals is a directional interval representation of~$G$.
	First, observe that any two intervals intersect if and only if they have a common clique.
	Next, if there is an arc \arc{u}{v} in~$G$, then the arc $\arc{u}{v}$
	is also in~$D$, $u < v$ holds both in~$L$ and in~$R$, the
	corresponding intervals overlap, and $\mathcal{I}(u)$ starts and ends
	to the left of $\mathcal{I}(v)$.
	Last, if there is an edge \edge{u}{v} in~$G$, then $u$ and~$v$ are independent in~$D$, $u < v$ in one of the orders in the realizer, and $v < u$ in the other.
	Thus, one of the intervals $\mathcal{I}(u)$ and $\mathcal{I}(v)$
	must contain the other.
\end{proof}

\cref{thm:recognition} follows easily from \cref{lem:recognition_rotation,lem:recognition_dimension}.
See~\cref{sec:implementation} for details.

\section{Open Problems}
\label{sec:outro}

Can we recognize directional interval graphs in linear time?
Can we recognize bidirectional interval graphs in polynomial time?
Can we color bidirectional interval graphs optimally, or at least find
$\alpha$-approximate solutions with $\alpha < 2$?

\bibliographystyle{abbrvurl}
\bibliography{paper-long-proceedings}

\newpage
\appendix

\section{Speeding Up the Greedy Coloring Algorithm}
\label{sec:runtimegreedy}

\RuntimeGreedy*
\label{lem:runtimegreedy*}

\begin{proof}
  We describe a sweep-line algorithm sweeping from left to right.
  In a first step, we show how to achieve a running time of $O(m + n
  \log n)$, where $m$ is the number of edges of the directional
  interval graph~$G$ induced by the given set~$V$ of $n$ intervals.
  Then we use an additional data structure in order to avoid the
  $O(m)$ term in the running time.  Note that $m$ can be quadratic
  in~$n$.  For the faster implementation, we do not assume knowledge
  of~$G$. 

  Build a balanced binary search tree~\T to keep track of the
  currently available colors.  Initially, \T contains the colors~1
  to~$n$.  Fill a list~\List with the $2n$ endpoints of the intervals
  in~$V$ (which we can assume to be pairwise different).  Sort~\List.
  Then traverse~\List in this order, which corresponds to a left-to-right sweep.
  There are two types of events.
  \begin{description}
  \item[\normalfont\textsc{Left:}] If the current endpoint is the left
    endpoint of an interval~$v$, let $x$ be the largest color over all
    intervals that have an arc to~$v$, that is,
    $x = \max\{c(v) \colon \arc{u}{v} \in A(G)\} \cup \{0\}$.  Then
    search in \T for the smallest color~$y$ greater than~$x$,
    delete~$y$ from~\T, and set $c(v)=y$.
  \item[\normalfont\textsc{Right:}] If the current endpoint is the
    right endpoint of an interval~$v$, we insert~$c(v)$ into~\T
    because $c(v)$ is available again.
  \end{description}
  
  Clearly, this implementation runs in $O(m + n \log n)$ time.  To
  avoid the $O(m)$ term, we use a second binary search tree $\T'$ that
  maintains the currently active intervals, sorted according to color.
  We augment~$\T'$ by storing, in every node~$\nu$, the leftmost right
  endpoint~$r^\nu$ in its subtree.  Any interval that contains the
  current endpoint in~\List is \emph{active}.

  At a \textsc{Left} event, this allows us to determine, in
  $O(\log n)$ time, the interval~$u$ with the largest color~$x$ among
  all active intervals that overlap the new interval~$v$ (that is,
  $r_u < r_v$), as follows.  We find~$u$ by descending~$\T'$ from its
  root.  From the current node, we go to its right child~$\rho$
  whenever $r^\rho<r_v$.  If such an interval does not exist, we set
  $x=0$.  Then we continue as above, querying~\T for the smallest
  available color~$y>x$.  Finally, we set $c(v)=y$ and add~$v$ to~$\T'$.

  At a \textsc{Right} event, we update~\T as above.  Additionally, we
  need to update~$\T'$.  We do this by deleting the interval~$v$ that
  is about to end.
  
  We now argue that, for outputting the greedy solution of our
  coloring problem, the running time of $\Oh{n \log n}$ is worst-case
  optimal assuming the comparison-based model of computation.
  Suppose that a coloring algorithm would run in $o(n \log n)$ time.
  Then, we could use it to sort any set $\{a_1,\dots,a_n\}$ of $n$
  numbers in $o(n \log n)$ time by coloring the set
  $\{[a_1-M,a_1+M],\dots,[a_n-M,a_n+M]\}$ of intervals,
  where~$M = \max\{a_1,\dots,a_n\} - \min \{a_1,\dots,a_n\}$.  Namely,
  the corresponding directional interval graph is a tournament graph
  and for each $i \in \{1,\dots,n\}$, the color of the interval
  $[a_i-M,a_i+M]$ in an optimal coloring corresponds to the rank
  of~$a_i$ in a sorted version of $\{a_1,\dots,a_n\}$.
\end{proof}

\Approx*
\label{cor:approx*}
\begin{proof}
        Let $\mathcal{I}$ be the set of intervals of~$G$.
	We split $\mathcal{I}$ into a set of left-going intervals $\mathcal{I}_1$
	and into a set of right-going intervals $\mathcal{I}_2$.
	These sets induce the directional graphs $G_1$ and $G_2$, respectively.
	Now we color $G_1$ and $G_2$ independently with our greedy coloring algorithm
	and we re-combine them by using different sets of colors for $G_1$ and $G_2$.
	This is a proper coloring of~$G$ with $\chi = \chi(G_1) + \chi(G_2)$ colors
	since between any interval in~$\mathcal{I}_1$ and any interval
        in~$\mathcal{I}_2$, there may be an edge but no arc.
	Clearly, $\chi \le 2 \max \set{\chi(G_1), \chi(G_2)} \le 2 \chi(G)$.
\end{proof}

\section{Coloring Mixed Proper Interval Graphs}
\label{sec:proper}

\HardnessPropper*
\label{clm:hardnesspropper*}

\begin{proof}
	The general idea is as follows.
	We start the construction with the same set of intervals as in the proof of \cref{thm:hardness}.
	Then, we set $x_\mathsf{left} = 0$, and $x_\mathsf{right}$ to the very right of all intervals, i.e., $x_\mathsf{right} = 4 (m+1)$.
	Next, we describe a procedure that extends every interval so that it has the left endpoint in $x_\mathsf{left}$, or has the right endpoint in $x_\mathsf{right}$.
	The procedure adds some new intervals and merges some groups of intervals into one interval.
	The total height of the interval representation increases to $4n+2nm$ during the procedure.
	Finally, we extend every interval at $x_\mathsf{left}$ ($x_\mathsf{right}$)	to the left (right) by a length inverse to its current total length.
	This trick guarantees that in the end, no interval contains another interval.
	In the remainder of the proof, we describe the procedure of extending, adding and merging intervals.
	
	The intervals of the frame and all $v_i^\mathsf{true}$ and $v_i^\mathsf{false}$
	with $i = 1, \dots, n$ already end at $x_\mathsf{left}$ or $x_\mathsf{right}$.
	Currently, we have that in any drawing of $G_\Phi$ with $6n$ layers and a fixed $i \in \set{1,\ldots,n}$, 
	all the intervals $b_i^j$, and $o_i^j$ with $j = 1, \dots, m$ are drawn in the layers of $f_i^3$ and $f_i^6$.
	Additionally, each dummy interval and each $s_j$ is draw in one of these layers.
	We divide these layers into $m$ copies each so that each pair of $b_i^j$ and $o_i^j$ has its own two layers.
	
	First we divide each $f_i^3$ and $f_i^6$ into $m$ copies each.
	Accordingly, we adjust the height of the drawing to be $4n + 2nm$.
	Then, we make $m$ copies of each dummy interval and virtually assign each copy to a distinct layer of the drawing.
	For each $b_i^j$ we virtually assign it to the layer of the $j$-th copy of $f_i^3$ and extend it to the left up to $x_\mathsf{left}$.
	In this process, we merge $b_i^j$ with every dummy interval on the left and with the $j$-th copy of $f_i^3$ while keeping all involved arcs.
	We call the merged interval $f_i^{3,j}$.
	If there is no $b_i^j$, we obtain $f_i^{3,j}$ by extending
	the $j$-th copy of $f_i^3$ up to $x_\mathsf{right}$	and merging
	it with all dummy intervals virtually assigned to its layer.
	
	Symmetric to $b_i^j$, we extend each $o_i^j$ to the right up to $x_\mathsf{right}$ and merge $o_i^j$ with all dummy intervals virtually assigned to the layer of $f_i^{3,j}$, but here we drop the arcs of the dummy intervals.
	We call the merged interval~${o'}_i^j$.
	Similarly, for every clause $c_j$ with variables $v_i$, $v_k$, $v_\ell$, we merge all dummy intervals virtually assigned to the layer of the $j$-th copy of $f_z^6$, for $z=i,k,\ell$ that are to the right of $s_j$ and drop all the arcs as in the previous case.
	We obtain three copies~${d'}_j^1$, ${d'}_j^2$, and ${d'}_j^3$ of the same interval and we merge one of these copies, say ${d'}_j^3$, with $s_j$.
	We denote that new interval by ${s'}_j$.
	We drop all arcs of ${d'}_j^1$, ${d'}_j^2$, and ${s'}_j$
	to preserve the freedom we had for placing $s_j$ in our original construction.
	The only unmerged dummy intervals are in the layer of the $j$-th copy of $f_i^6$ to the left of~$s_j$
	or in the layer of the $j$-th copy of $f_i^6$ if there is
	no occurrence of the variable $v_i$	in the clause $c_j$.
	In each of these layers, we merge the dummy intervals together
	with the corresponding copy of $f_i^6$ and
	obtain intervals ending at $x_\mathsf{left}$.	
	For $j = 1, \dots, m$, we call the merged interval $f_i^{6,j}$.
	
	For $i = 1, \dots, n$ and $j = 1, \dots, m-1$,
	we add the arcs $\arc{f_i^{2}}{f_i^{3,1}}$,
	$\arc{f_i^{3,j}}{f_i^{3,j+1}}$, $\arc{f_i^{3,m}}{f_i^{4}}$,
	$\arc{f_i^{5}}{f_i^{6,1}}$, $\arc{f_i^{6,j}}{f_i^{6,j+1}}$,
	and $\arc{f_i^{6,m}}{f_{i+1}^{1}}$ to have a
	frame as in the original hardness construction.
	Observe that this new frame now has exactly $4n + 2nm$ intervals with
	their left endpoint at~$x_\mathsf{left}$ and, in the whole construction,
	there are $2n + 6m$ other intervals with their right endpoint at~$x_\mathsf{right}$, i.e.,
	the $2n$ intervals $v_i^\mathsf{true}$ and $v_i^\mathsf{false}$ for $i = 1, \dots, n$
	and the $6m$ intervals ${o'}_i^j$, ${d'}_j^1$, ${d'}_j^2$, and ${s'}_j$
	for $j = 1, \dots, m$.
	
	Next, we argue that the functionality described
	in the proof of \cref{thm:hardness} is retained.
	Intervals of the (new) frame either block
	a complete layer from $x_\mathsf{left}$ to $x_\mathsf{right}$
	or they end at position 1 (each $f_i^2$ and $f_i^4$)
	or within the clause gadget of a variable $c_j$
	if a variable $v_i$ occurs in~$c_j$ (each $f_i^{3,j}$ and $f_i^{6,j}$).
	Any other interval starting in a clause gadget of a clause~$c_j$
	needs to be matched with a frame interval that ends in the clause
	gadget of~$c_j$.
	Therefore, to have a construction with a total height of at most~$4n + 2nm$,
	we need to combine $f_i^{3,j}$ and $f_i^{6,j}$
	with ${o'}_i^j$ and some of $\set{{d'}_j^1, {d'}_j^2, {s'}_j}$,
	while $f_i^{3,j}$ and ${s'}_j$ are not combinable.
	This ensures that the correctness argument from the
	proof of \cref{thm:hardness} remains valid.
\end{proof}

\section{Rotation Is Directional}
\label{sec:rotation}
\begin{claim}
	The rotation of \MPQ-tree $T$ of a directional interval graph $G$ constructed in the proof of \cref{lem:recognition_rotation} is a directional rotation.
\end{claim}
Let~$T'$ denote the tree~$T$ after applying the rotations described in the proof of \cref{lem:recognition_rotation}.
We claim that~$T'$
is directional.  Let~$\tilde T$ be an arbitrary directional rotation
of~$T$.  By construction, $T'$ and~$\tilde T$ differ only in the
ordering of children of \PQP-nodes~$x$ that do not have arcs from/to
vertices in~$A_x^T$.  To prove that~$T'$ is directional, it suffices
to show that the rotation of~$\tilde T$ obtained by swapping two
children of a P-node~$x$ that have no arcs from/to vertices in
$A_x^T$ is directional.

Consider a directional interval representation~$\mathcal I$ whose
clique ordering corresponds to the rotation~$\tilde T$ and
let~$y_k,y_l$ be two children of some \PQP-node~$x$ such that
neither~$B_k$ nor~$B_l$ contains a vertex with an arc from/to a
vertex in~$A_x^T$.  Let~$I_k$ be the smallest interval that
contains~$\mathcal I(v)$ for all~$v \in B_k$ and let~$I_l$ be
defined analogously for~$B_l$.  Note that~$I_k$ and~$I_l$ are
disjoint and that each of them is properly contained
in~$\bigcap_{v \in A_x^T} \mathcal I(v)$ as otherwise it would have
incoming or outgoing arcs.  After suitably stretching the real line,
we may assume that~$I_k$ and~$I_l$ have the same length.
Let~$x_k,x_l$ denote the left endpoints of~$I_k$ and~$I_l$,
respectively.  We obtain a directional representation whose clique
ordering corresponds to~$T'$ simply by exchanging the positions of
the representations of the subgraphs induced by~$B_l$ and by~$B_k$.
More formally, for each~$v \in B_k$
set~$\mathcal I'(v) = \mathcal I(v) - x_k+x_l$ and for
each~$v \in B_l$ set~$\mathcal I'(v) = \mathcal I(v) - x_l + x_k$.
For all other vertices $v \in V \setminus (B_k \cup B_l)$
set~$\mathcal I'(v) = \mathcal I(v)$.  It follows that each of the
subgraphs induced by~$B_k$ and~$B_l$ is still represented correctly.
Moreover, by construction the vertices in~$B_l$ and~$B_k$ still have
edges (and not arcs) to all vertices in~$A_x^T$.

\section{Recognition Algorithm}
\label{sec:implementation}
\ThmRecognition*
\label{thm:recognition*}
\begin{proof}
	Our algorithm, given a directional interval graph~$G$, applies the
	algorithm from Lemma~\ref{lem:recognition_rotation} to obtain a
	directional \MPQ-tree of~$G$.
	Then, using Lemma~\ref{lem:recognition_dimension}, it constructs a directional representation of~$G$.
	If any of the phases fails, then we know that~$G$ is not a directional interval graph, and we can reject the input.
	Otherwise, our algorithm accepts the input and produces the directional representation of~$G$.
	It is easy to see that both algorithms from
	Lemmas~\ref{lem:recognition_rotation}
	and~\ref{lem:recognition_dimension} can be implemented to run in
	$\Oh{|V(G)|^2}$ time.
	For Lemma~\ref{lem:recognition_rotation} it is enough to notice that:
	\begin{itemize}
		\item the \MPQ-tree of an interval graph $U(G)$ is of size $\Oh{|V(G)|}$ and can be constructed in time $\Oh{|V(G)|+|E(G)+A(G)|}$~\cite{KorteM89},
		\item when deciding the rotation of a \PQQ-node $x$, the pair of vertices that decide the rotation of $x$ can be found in $\Oh{|V(G)|}$ time,
		\item when deciding the rotation of a \PQP-node $x$, the first, and the last child of $x$ can be found in $\Oh{|V(G)|}$ time.
	\end{itemize}
	For Lemma~\ref{lem:recognition_dimension} it is enough to notice that:
	\begin{itemize}
		\item there are $\Oh{|V(G)|}$ maximal cliques in an interval graph,
		\item there are $\Oh{|V(G)|}$ vertices and $\Oh{|V(G)|^2}$ arcs in the auxiliary poset~$D$,
		\item two-dimensional realizer of the auxiliary poset~$D$ can be constructed in time $\Oh{|V(D)+A(D)|}$~\cite{McConnellS99}.
	\end{itemize}
\end{proof}

It is also quite easy to speed up the implementation of the rotation
algorithm in Lemma~\ref{lem:recognition_rotation} to linear time. 
Sadly, the auxiliary poset~$D$ that we construct in
Lemma~\ref{lem:recognition_dimension} has quadratic size and is thus the
main obstacle for obtaining a linear-time algorithm.
We suspect that an explicit construction of~$D$ can be avoided.

\end{document}